\documentclass[11pt]{article}
\usepackage[usenames,dvipsnames,svgnames,table]{xcolor}
\usepackage{enumitem}
\usepackage{graphicx}
\usepackage{fancybox}
\usepackage{comment}
\usepackage{xcolor}
\usepackage{nameref}

\definecolor{ForestGreen}{rgb}{0.1333,0.5451,0.1333}
\definecolor{DarkRed}{rgb}{0.8,0,0}
\definecolor{Red}{rgb}{1,0,0}
\usepackage[linktocpage=true,
pagebackref=true,colorlinks,
linkcolor=ForestGreen,citecolor=ForestGreen,
bookmarks,bookmarksopen,bookmarksnumbered]
{hyperref}

\usepackage{amsmath}
\usepackage{amssymb}
\usepackage{amsthm}

\usepackage[ruled, noend, linesnumbered]{algorithm2e}

\usepackage{subfig}

\usepackage{thm-restate}

\usepackage{float}
\usepackage{cleveref}

\newtheorem{theorem}{Theorem}[section]

\newtheorem{lemma}[theorem]{Lemma}

\newtheorem{claim}[theorem]{Claim}

\newtheorem{definition}[theorem]{Definition}

\newtheorem*{theorem*}{Theorem}
\newtheorem*{corollary*}{Corollary}
\newtheorem*{conjecture*}{Conjecture}
\newtheorem*{lemma*}{Lemma}
\newtheorem*{thm*}{Theorem}
\newtheorem*{prop*}{Proposition}
\newtheorem*{obs*}{Observation}
\newtheorem*{definition*}{Definition}

\newtheorem*{remark*}{Remark}
\newtheorem*{rec*}{Recommendation}

\newenvironment{fminipage}%
  {\begin{Sbox}\begin{minipage}}%
  {\end{minipage}\end{Sbox}\fbox{\TheSbox}}

\renewcommand{\deg}{\mathbf{deg}}
\newcommand{\vol}{\mathbf{deg}}
\newcommand{\ex}{\mathbf{ex}}

\def\defeq{\stackrel{\mathrm{def}}{=}}

\def\norm#1{\left\| #1 \right\|}

\newcommand\DDelta{\boldsymbol{\mathit{\Delta}}}
\newcommand\nnabla{\boldsymbol{\mathit{\nabla}}}

\newcommand\cc{\boldsymbol{\mathit{c}}}

\newcommand\ff{\boldsymbol{\mathit{f}}}

\def\tt{\boldsymbol{\mathit{t}}}

\newcommand\xx{\boldsymbol{\mathit{x}}}

\newcommand\veczero{\boldsymbol{0}}
\newcommand\vecone{\boldsymbol{1}}

\newcommand\BB{\boldsymbol{\mathit{B}}}

\newcommand\R{\mathbb{R}}

\newcommand{\polylog}{\text{ polylog}}

\usepackage{fullpage}

\newcommand*\samethanks[1][\value{footnote}]{\footnotemark[#1]}

\title{Parallel and Distributed Expander Decomposition: Simple, Fast, and Near-Optimal}
\author{Daoyuan Chen \\ ETH Zurich \\ chenda@student.ethz.ch \and Simon Meierhans\thanks{The research leading to these results has received funding from grant no. 200021 204787 of the Swiss National Science Foundation. Simon Meierhans is supported by a Google PhD Fellowship.} \\ ETH Zurich \\ mesimon@inf.ethz.ch \and Maximilian Probst Gutenberg\samethanks[1] \\ ETH Zurich \\ maximilian.probst@inf.ethz.ch \and Thatchaphol Saranurak\thanks{Supported by NSF grant CCF-2238138.} \\ University of Michigan \\ 
thsa@umich.edu}
\date{}

\begin{document}
\maketitle
\begin{abstract}
Expander decompositions have become one of the central frameworks in the design of fast algorithms. For an undirected graph $G=(V,E)$, a near-optimal $\phi$-expander decomposition is a partition $V_1, V_2, \ldots, V_k$ of the vertex set $V$ where each subgraph $G[V_i]$ is a $\phi$-expander, and only an $\widetilde{O}(\phi)$-fraction of the edges cross between partition sets.

In this article, we give the first near-optimal \emph{parallel} algorithm to compute $\phi$-expander decompositions in near-linear work $\widetilde{O}(m/\phi^2)$ and near-constant span $\widetilde{O}(1/\phi^4)$. Our algorithm is very simple and likely practical. Our algorithm can also be implemented in the distributed Congest model in $\tilde{O}(1/\phi^4)$ rounds.

Our results surpass the theoretical guarantees of the current state-of-the-art parallel algorithms \cite{chang2019improved, chang2020deterministic}, while being the first to ensure that only an $\tilde{O}(\phi)$ fraction of edges cross between partition sets. In contrast, previous algorithms \cite{chang2019improved, chang2020deterministic} admit at least an $O(\phi^{1/3})$ fraction of crossing edges, a polynomial loss in quality inherent to their random-walk-based techniques. Our algorithm, instead, leverages flow-based techniques and extends the popular sequential algorithm presented in \cite{saranurak2019expander}.
\end{abstract}

\section{Introduction}

Over the past two decades, expander decompositions have emerged as a central framework in graph algorithms. At a high level, these decompositions partition a graph into a collection of well-conditioned subgraphs and a small set of crossing edges. Formally, a $\phi$-expander $G$ is a graph in which, for any subset $S \subseteq V$, the number of edges leaving $S$ is large compared to the volume of the smaller side of the cut $(S, V \setminus S)$. Specifically, $|E_G(S, V \setminus S)| \geq \phi \cdot \min\{\vol_G(S), \vol_G(V \setminus S)\}$. A $\phi$-expander decomposition of a graph $G = (V, E)$ is a partition of the vertex set $V$ into clusters $V_1, V_2, \ldots, V_k$ such that each subgraph $G[V_i]$ is a $\phi$-expander. The error of an expander decomposition is the number of edges in $E$ that cross between clusters $V_i$ and $V_j$ for $i \neq j$. Intuitively, $\phi$-expanders are well-connected clusters where the 'well-connectedness' increases with $\phi$. When $\phi \approx 1$, the graph is very well-connected, has low diameter, and exhibits spectral properties similar to those of a complete graph. Conversely, every simple connected graph is a $\phi = 1/|V|^2$ expander.

Expander decompositions have been pivotal in developing the first deterministic and randomized almost-linear time algorithms for various fundamental graph problems such as maximum and min-cost flow \cite{kelner2014almost, chen2022maximum}, electrical flows \cite{spielman2004nearly}, Gomory-Hu trees \cite{abboud2023all}, and edge- and vertex-connectivity \cite{kawarabayashi2018deterministic, li2021deterministic,li2021vertex}, and many more \cite{wulff2017fully, nanongkai2017dynamic,nanongkai2017dynamicMinimum,chuzhoy2019new,bernstein2020deterministic,bernstein2020fully,chuzhoy2021deterministic, goranci2021expander, chuzhoy2021decremental,bernstein2022deterministic, bernstein2022deterministic, jin2022fully, kyng2023dynamic, chen2023almost, jin2024fully, chuzhoy2024maximum, vdB2024decrMincost}.

This research program seeking fast almost-linear time algorithms is driven by the increasing volume of data, resulting in large-scale graphs for which it is prohibitive to spend more than almost-linear time relative to the input size.

In a first effort to translate these new almost-linear time algorithms from theory to practice, \cite{practicalExpDecomp} gave a first practical implementation of an algorithm to compute expander decompositions for the important regime of $\phi = \tilde{\Omega}(1)$\footnote{In this article we use $\tilde{O}(\cdot)$ and $\tilde{\Omega}(\cdot)$ to hide $\polylog |V|$ and $1/\polylog |V|$ factors respectively.}. These implementations demonstrate that expander decompositions can be computed in reasonable time (under or roughly one day on a high-performance computer) for medium-sized graphs with roughly 400,000 to 35,000,000 edges.\footnote{In their experiments, they used a 2x8-core Intel Xeon Gold 6144 Skylake CPU, clocked at 3.5GHz with 24.75MB L3 cache and 192GB DDR4 RAM (2666 MHz).} This is achieved by implementing the framework suggested in \cite{saranurak2019expander}, the simplest and theoretically fastest algorithm to compute expander decompositions, augmented by various practical optimizations and heuristics.

While a promising first step, the result from \cite{practicalExpDecomp} strongly suggests that either the algorithm has to be significantly improved, or that parallelization has to be harnessed in order to be able to handle large-sized graphs, meaning graphs that are of the order of a hundred-million or even a billion edges.

\subsection{Our Contribution}

In this paper, we present a new parallel algorithm to compute expander decompositions. 

\begin{theorem}[Parallel Expander Decomposition]\label{thm:exp-decomp}\label{thm:mainResult}
Given a graph $G=(V,E)$ of $m$ edges and a parameter $\phi\in(0,1)$, there is a randomized parallel algorithm that with high probability\footnote{In this article we say with high probability to mean that for every constant $C > 0$, there is such an algorithm that succeeds with probability at least $1-n^{-C}$.} finds a $\phi$-expander decomposition $V_1,..,V_k$ with error $\sum_{i 
< j} |E_G(V_i, V_j)|=\widetilde{O}(\phi m)$. The total work of the algorithm is $\widetilde{O}(m/\phi^2)$ with span $\widetilde{O}(1/\phi^4)$.
\end{theorem}

Our algorithm is also implementable in the distributed \textsc{Congest} model. 

\begin{theorem}[Distributed Expander Decomposition]
Given a graph $G=(V,E)$ of $m$ edges and a parameter $\phi\in(0,1)$, there is a randomized distributed algorithm in the \textsc{Congest} model that with high probability finds a $\phi$-expander decomposition $V_1,..,V_k$ with error $\sum_{i 
< j} |E_G(V_i, V_j)|=\widetilde{O}(\phi m)$. The number of rounds are $\widetilde{O}(1/\phi^4)$.
\end{theorem}

Distributed expander decompositions are useful in particular because many parallel algorithms can be directly ported to the distributed model on expander graphs \cite{10.1145/3087801.3087827, ghaffari_et_al:LIPIcs.DISC.2018.31, chang17distexp}. In this article, we only prove the parallel expander decomposition result. The algorithm can be ported to the distributed model directly. 

Our algorithm yields expander decompositions of nearly-optimal error in quality and has extremely low span for the important setting where $\phi = \tilde{\Omega}(1)$. At the same time, our algorithm is still very simple and likely practical to the extent that it can most likely be integrated with the implementation from \cite{practicalExpDecomp} with relatively little overhead as it extends the framework of \cite{saranurak2019expander}. 

Our result also simultaneously surpasses all theoretical guarantees obtained by previous parallel and distributed algorithms \cite{chang2019improved, chang2020deterministic}. In particular, our algorithm is the first to obtain near-optimal error for any value of $\phi$ improving from the currently best error of $\tilde{O}(\phi^{1/3} m)$ achieved by \cite{chang2020deterministic}. Our algorithm matches the work bound of the best sequential algorithm \cite{saranurak2019expander} up to a $\tilde{O}(1/\phi)$ factor. We refer the reader to \Cref{sec:review_exp_decomp} for a review of sequential, parallel and distributed expander decomposition algorithms. 

We further believe that our newly developed techniques are of general interest in the design of other parallel algorithms in general and in the design of faster, stronger and more robust sequential and dynamic algorithms for expander decompositions in particular. 

\subsection{Our Techniques}

Our algorithm takes the framework of \cite{saranurak2019expander} as a starting point. In their framework, the popular cut-matching algorithm \cite{khandekar2009graph} is used to find balanced sparse cuts, that is cuts $(S, V \setminus S)$ where $|E_G(S, V \setminus S)| < \phi \cdot \min\{\vol_G(S), \vol_G(V \setminus S)\}$ and the degree incident to $S$ and $V \setminus S$ is $\tilde{\Omega}(m)$, respectively; or certify that no such balanced sparse cut exists. If a sparse cut is found, it is used to refine the partition of the vertex set. 

Conversely, if no balanced sparse cut exists the algorithm returns a vertex set $A \subseteq V$ of large volume such that $G[A]$ is a nearly expander. Formally, we have $\deg_G(V \setminus A) = O(m/\polylog(m))$ and every $S \subseteq A$, $|E_G(S, V \setminus S)| \geq \phi \cdot \min\{\vol_G(S), \vol_G(V \setminus S)\}$. From an algorithmic perspective, it would be desirable to turn $A$ into a partition set and recurse on $V \setminus A$, but unfortunately $G[A]$ isn't quite a $\phi$-expander. Therefore, a trimming algorithm is used to find a large subset $A' \subseteq A$ such that $G[A']$ is a $\Omega(\phi)$-expander, and the algorithm recurses on $V \setminus A'$ instead.

\paragraph{Trimming with Few Iterations.} In \cite{saranurak2019expander}, a trimming algorithm is obtained by constructing a flow problem on $G[A]$ where every edge has capacity $2/\phi$, every vertex $v$ is a source with $2/\phi \cdot (\deg_G(v) - \deg_{G[A]}(v))$ mass and $\deg_G(v)$ sink capacity. It is then shown that the min-cut in this flow problem yields a relatively sparse cut in $G[A]$ and that the larger side $A'$ of this cut is a $\Omega(\phi)$-expander. 

Since current fast maximum flow algorithms are far from practical and have large sub-polynomial overhead, the state-of-the-art trimming algorithms use the unit flow algorithm from \cite{henzinger2020local}, a height-constrained version of the famous push-relabel algorithm, that heavily exploits the uniform sinks of the flow problem above. They set up said flow problem for sets $A_0 = A, A_1, \ldots, A_k$ where $A_{i+1} \subseteq A_i$. The $i$-th flow problem is then used to determine which vertices need to be removed to arrive at $A_{i+1}$ and eventually the algorithm will solve the flow problem and decide that $G[A_k]$ is a $\Omega(\phi)$-expander. 

But while the authors cleverly re-use information across the flow problems to obtain fast runtime, the algorithm is inherently sequential as it needs to solve the $i$-th flow problem before it even knows $A_{i+1}$. In our work, we observe that a rather simple technique can be used to bound the number of flow problems that ever need to be solved by $\log_2 n$. To do so, we show that it suffices to grow the sink capacity over the flow problems slowly, i.e. instead of directly admitting $\deg_G(v)$ sink capacity, our algorithm admits only $\frac{i \cdot \deg_G(v)}{\log_2 n}$ sink capacity in the $i$-th flow problem. We show that this speeds up convergence exponentially. The technique is described and analyzed in detail in \Cref{sec:trimming}.

\paragraph{Parallel Unit Flow.} Our second contribution is to parallelize the unit flow algorithm that was given in \cite{henzinger2020local}. As a starting point, we exploit the simple insight from \cite{GT88PushRelabel} for the underlying push-relabel algorithm that one can work with rounds where in even rounds, levels remain fixed and all possible pushes are executed, and in odd rounds all vertices that still have excess can be relabelled in parallel. We note that via the analysis of \cite{GT88PushRelabel}, this only yields a trivial upper bound on the number of rounds.

However, we show that as long as at most half of the original source mass is settled (meaning absorbed or placed as excess at an 'inactive' vertex of maximum admissible level), we can carry out an $\tilde{\Omega}(1/\phi^2)$-fraction of the total work in every round. Further, by exploiting the structure of the flow problem, we show that we can implement each round with span $\tilde{O}(1/\phi)$. Thus, we can obtain a span of $\tilde{O}(1/\phi^3)$. 

Clearly, this technique only works until half the original source mass is settled. However, we again exploit that growing the sink capacity keeps the algorithm efficient. Concretely, each parallel unit flow problem has for every vertex $v$ a sink capacity of at least $\Delta(v) = \deg_G(v)/\log_2(n)$. We start solving the flow problem while only admitting $\frac{\Delta(v)}{8 \cdot \log_2(n)}$ capacity to each sink and then grow the capacity by $\frac{\Delta(v)}{8 \cdot \log_2(n)}$ whenever the unsettled source mass drops by a factor of $1/2$. This ensures that the span remains $\tilde{O}(1/\phi^3)$ for each flow problem, and therefore the total span is at most $\tilde{O}(\log n/\phi^3)$, as desired. 

Our parallel unit-flow algorithm is described and analyzed in \Cref{sec:parallelimplementation}.

\paragraph{A Distributed Algorithm. } Since we employ the distributed framework of \cite{chang2019improved, chang2020deterministic} for implementing the cut matching game, it suffices to realize that all our flow algorithms and trimming procedures can be implemented in the distributed \textsc{Congest} model as well. Our parallel unit flow algorithm is completely local since the only operations we perform is locally adding sink, and pushing one unit of flow across an edge. Our trimming procedure in \Cref{sec:trimming} uses parallel unit flow and ball growing to identify a sparse cut. Growing a ball to the adequate diameter can be implemented in the \textsc{Congest} model directly. Finally, the parallel matching algorithm used to implement the cut matching game in \Cref{sec:proofOfMain} uses a combination of our unit flow algorithm and ball growing, which is again implementable in a distributed fashion directly. 

\subsection{Roadmap}

For the rest of this article, we focus on obtaining a fast parallel trimming routine. The main technical result is stated in \Cref{lma:trimming}. We defer the proof of \Cref{thm:mainResult} to \Cref{sec:proofOfMain} since the parallelization of the framework of \cite{saranurak2019expander} is rather straightforward given a parallel trimming algorithm and has already been done previously in \cite{chang2020deterministic}.

\section{Preliminaries}

\paragraph{Graphs.} We denote graphs as tuples $G = (V, E)$ and refer to the number of vertices and edges with $n$ and $m$ respectively. For an undirected graph $G=(V,E)$ and a subset $A \subset V$, we use $G[A]$ to denote the induced graph and $E[A]$ to denote the set of edges with both endpoints in $A$, i.e. $E[A] = \{ \{u,v\} \in E, u,v \in A\}$. For disjoint subsets $A, B\subset V$, we denote by $E(A, B)$ edges in $E$ with exactly one endpoint in $A$ and $B$. We denote with $\deg_G$ the degree vector of graph $G$, and for a set $S \subseteq V$ we let $\deg_G(S) \defeq \sum_{v \in S} \deg_G(v)$. We let $\BB \in \R^{E \times V}$ denote the edge-vertex incidence matrix for an arbitrary but fixed orientation of the edges of graph $G$. 

We sometimes use $X_G$ to refer to a variable $X$ associated with graph $G$ to remove ambiguity. 

\paragraph{Flows.} A \textit{(residual) flow instance} $\Pi=(G, \cc, \ff, \DDelta, \nnabla)$ consists of an undirected graph $G=(V,E)$, a capacity vector $\cc \in \R_{\geq 0}^E$, a feasible flow  vector $\ff \in \R^E$, a source vector $\DDelta \in \R^V_{\geq 0}$ and a sink vector $\nnabla \in \R^V_{\geq 0}$. We use \textit{mass} to refer to the substance routed. 

For a vertex $v \in V$, $\DDelta(v)$ specifies the amount of mass initially placed on $v$, and $\nnabla(v)$ specifies the capacity of $v$ as a sink, i.e., the amount of mass that $v$ can absorb. For an edge $\cc(e)$ specifies the amount of mass that can be routed along $e$ in either direction. A flow $\ff$ is said to be feasible if $-\cc \leq \ff \leq \cc$. 

We say $e \in E$ is \textit{saturated} if $|\ff(e)|= \cc(e)$. Given a flow $\ff$ on the undirected graph $G$ we let $G_{\ff} = (V, E_{\ff})$ denote the residual graph where the edge $e = (u,v)$ has residual capacity $\cc(e) - \ff(e)$ in direction of $(u,v)$, and $\cc(e) + \ff(e)$ in direction $(v,u)$.
We use $E_{\ff}(A,B)$ to denote the edges in $E(A,B)$ with nonzero residual capacity going from $A$ to $B$. For convenience, we let $\ff(u,v) \defeq \ff(e)$ for $e = (u,v)$ and $\ff(v, u) \defeq - \ff(u,v)$ denote the flow in the other direction. Finally, we let $\cc_{\ff}$ be the directed residual capacity vector for initial capacities $\cc$ and a flow $\ff$.

For a given flow $\ff$ on a graph $G$ with source $\DDelta$ and sink $\nnabla$, we let $\ex^G_{\ff, \DDelta, \nnabla} \defeq \max(\BB_{G}^T\ff + \DDelta - \nnabla, \veczero)$ denote the excess (source) left over. 

\paragraph{Linear Algebra. } Given a vector $\xx \in \R^A$, we let $\xx[B]$ for $B \subseteq A$ denote the vector in $\R^B$ with $\xx[B](i) = \xx(i)$ for $i \in B$.

\section{The Trimming Step} \label{sec:trimming}

In this section, we present our parallel algorithm for trimming nearly expanders to expander graphs. Given a graph $G$, a set $A$ is said to induce a nearly expander if it expands in the context of the whole graph $G$. 

\begin{definition}[Nearly Expander]
    \label{def:nearly_exp}
    Given $G = (V, E)$ and $A \subseteq V$, $G[A]$ is a $\phi$-nearly expander in $G$ if for all $S \subseteq A$ such that $\vol_G(S) \leq \vol_G(A \setminus S)$, we have $|E_G(S, V \setminus S)| \geq \phi \cdot \vol_G(S)$. 
\end{definition}

Notice that \Cref{def:nearly_exp} would correspond to the definition of expander graphs if the edge set $E_G(S, V\setminus S)$ was restricted to edges internal to $G[A]$. Therefore, every induced expander is also a nearly-expander, but the contrary is not the case in general. Historically, nearly expanders play a crucial role in the development of efficient algorithms for computing expander decompositions, and the first near-linear time algorithms only guaranteed that each component be a nearly expander \cite{st:exp_dec}. To obtain the stronger and significantly easier to work with guarantee that each cluster be an expander graph, expander trimming algorithms output a set $A' \subseteq A$ such that $G[A']$ is a $\Omega(\phi)$-expander and $A'$ has nearly the same size as $A$. 

We state the main lemma we prove in this section.

\begin{restatable}[Parallel Trimming]{lemma}{parallelTrimming} \label{lma:trimming}
Given a graph $G = (V, E)$, a parameter $\phi \in \R_{\geq 0}$ and a set $A \subseteq V$ such that $G[A]$ is a $\phi$-nearly expander and $|E(A, V \setminus A)| \leq \phi \cdot m$ the algorithm $\textsc{Trimming}(G = (V, E), A, \phi)$ outputs a set $A'$ such that 
\begin{enumerate}
    \item $G[A']$ is a $\phi/6$ expander
    \item $\vol_G(A') \geq \vol_G(A) - \frac{4 \log^2 n }{\phi}|E(A, V \setminus A)|$
    \item $|E(A', V \setminus A')| \leq 2 \cdot |E(A, V \setminus A)|$
\end{enumerate}
in total work $\tilde{O}(m/\phi^2)$ and total span $\tilde{O}(1/\phi^3)$. 
\end{restatable}

\paragraph{Certifying Expansion via Flows. } In \cite{saranurak2019expander} (inspired by \cite{orecchia2014flow, nanongkai2017dynamic, nanongkai2017dynamicMinimum}), a flow problem that both guides the pruning and certifies expansion is introduced. We first observe that for a $\phi$-nearly expander $G[A]$, we have that $G[A']$ is also a $\phi$-nearly expander for all sets $A' \subseteq A$. 

\begin{lemma}[See Proposition 3.2 in \cite{saranurak2019expander}]
    \label{lem:flow_cert}
    Given a graph $G = (V, E)$ and a set $A \subseteq V$ such that $G[A]$ is a $\phi$-nearly expander and there exists a flow $f$ supported on the graph $G[A]$ that routes source $\Delta(v) \defeq 2/\phi\cdot(\deg_G(v) - \deg_{G[A]}(v))$ to sinks $\nabla(v) \defeq \deg_G(v)$ for $v \in A$ on $G[A]$ with uniform edge capacity $2/\phi$. Then, $G[A]$ is a $\frac{\phi}{6}$-expander. 
\end{lemma}
\begin{proof}
    Consider an arbitrary cut $S \subset A$ such that $\vol_{G}(S) \leq \vol_{G}(A \setminus S)$. We aim to show that the flow $f$ certifies $|E_{G[A]}(S, A \setminus S)| \geq \frac{\phi}{6}\cdot\vol_{G}(S)$. To show the lemma by contradiction, we assume
    \begin{equation}
        \label{eq:sparse_cut}
        |E_{G[A]}(S, A \setminus S)| < \frac{\phi}{6}\cdot\vol_{G}(S).
    \end{equation} Since the cut $S$ also satisfies the weaker condition $\vol_{G}(S) \leq \vol_{G}(V \setminus S)$, we have $E(S, V \setminus S) \geq \phi\cdot\vol_{G}(S)$. Therefore, we have that 
    \begin{align*}
        \sum_{v \in S} \Delta(v) &= \frac{2}{\phi}\cdot\sum_{v \in S} (\deg_G(v) - \deg_{G[A]}(v)) \\ &= \frac{2}{\phi}\cdot(|E(S, V \setminus S)| - |E_{G[A]}(S, A \setminus S)|) \\ &\stackrel{a)}{\geq} \frac{2}{\phi}\cdot(\phi\cdot\vol_G(S) - \frac{\phi}{6}\cdot \vol_G(S)) = \frac{5}{3} \cdot\vol_G(S)
    \end{align*} 
    where a) follows from our assumption \eqref{eq:sparse_cut} and \Cref{def:nearly_exp}. Therefore the total source mass originating inside $S$ is at least $\frac{5}{3}\cdot\vol_G(S)$. Since the total amount of source capacity inside $S$ is $\sum_{v \in S} \nabla(v) = \vol_G(S)$, at least $\frac{2}{3}\cdot\vol_G(S)$ units of flow have to be routed across the cut $E_{G[A]}(S, A \setminus S)$. Every edge has capacity $\frac{2}{\phi}$, and therefore there are at most $\frac{1}{3}\cdot\vol(S)$ units of flow can be routed by \eqref{eq:sparse_cut} which leads contradiction. This concludes the proof. 
\end{proof}

\subsection{The Algorithm}

For our trimming algorithm, we crucially rely on the following parallel implementation of the unit-flow algorithm suggested in \cite{HenzingerRW17localflow, saranurak2019expander} for structured flow problems. More precisely, we use the following result whose proof we defer to \Cref{sec:parallelimplementation}. 

\begin{restatable}{lemma}{parallelPushRelabel}
\label{lma:parallelPushRelabel}
Given a height parameter $h$ and a residual flow instance $\Pi = (G, \cc, \DDelta, \nnabla)$ where $\nnabla(v) \geq \gamma \cdot \deg(v)$ for all vertices $v \in V$ for some $0 < \gamma \leq 1$, $\norm{\DDelta}_1 \leq 2m$ and $\DDelta(v) \leq \eta \cdot \deg_G(v)$ for all $v \in V$, and $\|\cc\|_{\infty} \leq \eta$. 

Then, there is a parallel algorithm $\textsc{ParallelUnitFlow}(G, 
\cc, \DDelta, \nnabla, h)$, that requires work $\tilde{O}(mh\eta/\gamma)$ and span $\tilde{O}(h^2\eta/\gamma)$, and produces a flow $\ff$ and labeling $l:V\longrightarrow\{0,...,h\}$ such that:
\begin{enumerate}[label=(\roman*)]
    \item If $l(u)>l(v)+1$ where $\{u,v\}$ is an edge, then $\{u,v\}$ is saturated in the direction from $u$ to $v$, i.e. $\ff(u,v)=\cc(u,v)$.\label{prop:gvalidProp1}
    \item If $l(u)\geq 1$, then $u$'s sink is nearly saturated, i.e. $\ff(u)\geq \nnabla(u)/(8 \cdot \log_2 n)$.\label{prop:gvalidProp2}
    \item[(iii)] \label{prop:noExcessAtLowLevels} If $l(u)<h$, then there is no excess mass at $u$, i.e. $\ex_{\DDelta,\nnabla,\ff}^G(u) = 0$.
\end{enumerate}
\end{restatable}

Next, we describe how we obtain a parallel trimming algorithm from the $\textsc{ParallelUnitFlow}()$ subroutine. We refer the reader to detailed pseudo-code given in \Cref{alg:trimming}. This algorithm is given a graph $G = (V, E)$, a subset $A$ and a parameter $\phi$ such that $G[A]$ is a $\phi$-nearly expander. We let $A_0 \gets A$ and we give each edge in the graph $G[A] = (V, E)$ capacity $\frac{2}{\phi}$ throughout.

After initialization, the algorithm enters the main-loop. Each iteration $i$ of said loop produces a flow $\ff_i$. To compute the said flow, we call $\textsc{ParallelUnitFlow}()$ on the induced graph $G_{[A_{i - 1}]}$ with residual capacities $\cc_{\ff_{i-1}}$. We set the source and sink function to the excess of the previous iteration and $\frac{\deg_G(v)}{\log_2(n)}$ respectively. Then we, call $\textsc{ParallelUnitFlow}()$ on this instance with a total of $h \defeq \frac{5120}{\phi}  \cdot \log_2^2 n \cdot \ln m$ levels to obtain the flow $
\ff_i$. If the excess after this call is $0$ the algorithm returns the flow $\ff_i$ and the level vector $l_i$ and terminates. Otherwise, we aim to find a cut that further reduces excess. To do so, we initialize $S_0 = \{v \in A_{i-1}: {l}_i(v) \geq h\}$ to be the subset of $A_{i - 1}$ that still has excess flow to be routed. Then, in the $j$-th iteration of a sub-loop we check if $|E_{\ff_i}(S_j, A_{i-1} \setminus S_j)| \geq \frac{5\ln m}{h}\cdot \deg_{G}(S_j)$. If not, we let $A_i \gets A_{i - 1} \setminus S_j$ and continue with the next iteration of the main loop. Otherwise, we continue with $S_j = \{v \in A_{i-1}: {l}_i(v) \geq h - j\}$.

\begin{algorithm}[H]
\caption{\textsc{Trimming}$(G = (V,E), A, \phi)$} \label{alg:trimming}
\DontPrintSemicolon
 $h \defeq \frac{5120}{\phi} \ \cdot \log_2^2 n\cdot \ln m$\\
$\cc \gets \frac{2}{\phi}\cdot\vecone$\\
$A_0 \gets A$; $\ff_0 \gets \veczero$, $\DDelta_0 \gets \frac{2}{\phi}(\deg_G[A] - \deg_{G[A]})$; $\nnabla_0 \gets \veczero$; $i \gets 0$\\

\While(\label{lne:whileOuter}\tcc*[h]{While we do not have a feasible flow.}){\textbf{true}}{
    $i \gets i +1$\\
    $\nnabla_i \gets \nnabla_{i - 1} + \frac{1}{\log_2 n} \deg_G(v)[A_{i -1}]$ \\
    $({\ff'}_i, {l}_i) \gets \textsc{ParallelUnitFlow}(G[A_{i-1}], \cc_{
    \ff_{i - 1}}[A_{i-1}], \ex^{G[A_{i - i}]}_{\ff_{i - 1}, \DDelta_{i - 1}, \nnabla_{i - 1}}, \frac{\deg_G(v)[A_{i -1}]}{\log_2 n} , h)$ \label{lne:parallel_unit}\\
    $\ff_i \gets \ff_{i - 1} + \ff'_i$ \\
    \lIf(\label{lne:ifReturn}){$\ex^{G[A_i]}_{\ff_i, \DDelta_{i}, \nnabla_i} = \veczero$}{
        \Return $A' = A_{i-1}$
    }

    $j \gets 0$; $S_0 \gets \{v\in A_{i-1}: {l}_i(v)=h\}$

   \While{$|E_{\ff_i}(S_j, A_{i-1} \setminus S_j)| \geq \frac{5\ln m}{h}\cdot \deg_{G}(S_j)$}{ \label{lne:batch_pruning:while}
        $j \gets j + 1$; $S_j \gets \{v \in A_{i-1}: {l}_i(v) \geq h-j\}$ 
    }
    
    $A_i \gets A_{i-1} \setminus S_j$.\label{lne:induceByCut}  \\
    $\DDelta_i \gets \frac{2}{\phi}(\deg_G[A_i] - \deg_{G[A_i]})$; $\nnabla_i \gets \nnabla_i[A_i]$
}
\end{algorithm}

\subsection{Correctness}

In this section, we argue that $G[A']$ is a $\phi/6$-expander when the algorithm terminates. This almost immediately follows from \Cref{lem:flow_cert}. 

\begin{claim} \label{clm:expansion}
    If \Cref{alg:trimming} terminates for $i \leq \log_2(n)$, then $G[A']$ is a $\phi/6$-expander. 
\end{claim}
\begin{proof}
    If the algorithm terminates in iteration $i$, then the final flow $\ff_i$ routes the source $\deg_G(v) - \deg_{G[A']}(v)$ to sinks $\nnabla_i(v) = \frac{i \cdot \deg_G(v)}{\log_2 n} \leq \deg_G(v)$ for $v \in A'$ by the definition of our algorithm and the assumption that $i \leq \log_2 n$. Therefore, \Cref{lem:flow_cert} applies and we conclude that $G[A']$ is a $\phi/6$-expander. 
\end{proof}

In the remainder of this section, we show that the algorithm indeed terminates quickly, and that the set $A'$ is not much smaller than $A$. 

\subsection{Runtime}

We first argue that the calls to $\textsc{ParallelUnitFlow}()$ are of the desired form. 

\begin{claim} \label{clm:work_span_oracle}
    Whenever $\textsc{ParallelUnitFlow}()$ is called in \Cref{lne:parallel_unit} of \Cref{alg:trimming}, it requires work $\tilde{O}(m/\phi^2)$ and span $\tilde{O}(1/\phi^3)$. 
\end{claim}
\begin{proof}
    Since we add $\frac{1}{\log_2 n}\cdot \deg_G(v)$ sink before calling the algorithm, \Cref{lma:parallelPushRelabel} applies with $\gamma = 1/\log_2 n$. The desired work and span bounds follow from the definition of $h = \tilde{O}(\phi^{-1})$ and the fact that the residual capacities are never larger than $4/\phi$. 
\end{proof}

Then, we show that the while loop at \Cref{lne:batch_pruning:while} terminates in less than $h$ steps. 

\begin{claim}
    The while loop at \Cref{lne:batch_pruning:while} terminates in less than $h$ steps. 
    \label{clm:while_terminates}
\end{claim}
\begin{proof}
    Consider the $j$-th iteration of the while loop. We observe that by item 1 in \Cref{lma:parallelPushRelabel}, all edges leaving $S_j$ in the residual graph have both endpoints in the set $S_{j + 1}$. Therefore, we obtain that $\vol_G(S_{j + 1}) \geq \left(1 + \frac{5 \ln m}{h}\right)\vol_G(S_j)$ for every $j$ until the algorithm terminates. Assume for the sake of a contradiction that the algorithm reaches iterate $h$. Then we have that $\vol_G(S_h) \geq \left(1 + \frac{5 \ln m}{h}\right)^h$ since $\vol_G(S_0) \geq 1$ because the algorithm would have terminated otherwise. But $\left(1 + \frac{5 \ln m}{h}\right)^h \geq n^3$ which is a contradiction since it exceeds the volume of the graph $G$. Therefore, the while loop at \Cref{lne:batch_pruning:while} terminates in less than $h$ steps. 
\end{proof}

To conclude correctness and runtime analysis, we argue that the main loop at \Cref{lne:whileOuter} of the algorithm converges in $\log_2 n$ steps. To do so, we show that the remaining excess is reduced sufficiently in each round of the algorithm. 

\begin{claim}
    \label{clm:outer_while_term}
    The main loop at \Cref{lne:whileOuter} of \Cref{alg:trimming} terminates after at most $ \log_2 n$ steps.
\end{claim}
\begin{proof}
    The initial excess is at most $\frac{2}{\phi}\cdot m\phi = 2m \leq n^4$. We next show that the remaining excess is reduced by a factor $\frac{1}{32}$ in each iteration.

    Let's fix an iteration $i$, denote the total excess at the end of iteration $k$ by $X^k$, and let $S_j$ denote the final set the algorithm settles on when the while loop at \Cref{lne:batch_pruning:while} terminates. We have $j < h$ by \Cref{clm:while_terminates}. Furthermore, by \Cref{lma:parallelPushRelabel} all vertices in $S_j$ have nearly saturated sinks, and since we set the sink to $\deg_G(v)/\log_2 n$ for every vertex $v \in A_{i - 1}$ we obtain that the total excess $X^{i - 1}$ at the end of the previous iteration was at least
    \begin{equation}
        \label{eq:lower}
        X^{i - 1} \geq \vol_G(S_j)/(8\log_2^2 n)
    \end{equation}

    By the termination condition, we obtain that the total number of edges between $S_j$ and $A_{i - 1} \setminus S_j$ in the residual graph are at most $\frac{5\ln m}{h} \cdot \deg_G(S_j)$. Only these edges can contribute residual demand, since the flow that the other edges add is already routed away. Each such edge can add at most $4/\phi$ units of flow. Therefore, we have that $X^i \leq \frac{4}{\phi}\cdot\frac{5\ln m}{h} \cdot \deg_G(S_j) = \frac{1}{256 \log_2^2 n}\cdot\deg_G(S_j)$. Together with \eqref{eq:lower} we obtain that $X^i \leq X^{i - 1}/32$. Therefore, the algorithm terminates after $\log n$ iterations, since the total remaining excess at this point is at most $n^4/32^{\log_2 n} = n^4/2^{5 \log_2 n} = 1/n < 1$. Since all the flows we consider are integral, this means that the flow routes the demand and therefore the loop terminates. This concludes the proof of this claim. 
\end{proof}

\subsection{Proof of \Cref{lma:trimming}}

Before we conclude with the proof of \Cref{lma:trimming}, we show that the set $A'$ returned by the algorithm is still large as a function of $E(A, V \setminus A)$.

\begin{claim} \label{clm:size_loss}
    $\vol_G(A') \geq \vol_G(A) - \frac{4 \log^2 n}{\phi}|E(A, V \setminus A)|$
\end{claim}
\begin{proof}
    The proof follows the template of the proof of \Cref{clm:outer_while_term}. Every vertex $v$ in $A \setminus A'$ absorbs at least $\deg_G(v)/(8\log_2^2 n)$ flow. Initially, the total source mass is exactly $\frac{2}{\phi}\cdot|E(A, V \setminus A)|$.  Whenever we remove a set $S_j$ from $A$ we introduce at most $\frac{4}{\phi} \cdot \frac{5 \ln m }{h} \cdot\vol(S_j) = \frac{1}{256 \log_2^2 n} \cdot\vol_G(S_j)$ extra source flow. But at the same time, at least $\vol_G(S_j)/(8 \cdot \log_2^2 n)$ flow was absorbed within the removed part $S_j$. Therefore, the charged amount of flow in the graph got reduced by a factor at least $1/4$. The claim follows from a geometric sum. 
\end{proof}

\begin{claim} \label{clm:cut_sparse}
    $|E(A', V \setminus A')| \leq 2 \cdot |E(A, V \setminus A)|$
\end{claim}

\begin{proof}
    We first notice that the total flow that is initially in the system is exactly $\frac{2}{\phi} \cdot |E(A, V \setminus A)|$. We now show that the total number of new crossing edges when going from sets $A_{i - 1}$ to $A_i$ is small. If an edge between $S_j$ and $A_{i - 1} \setminus S_j$ carries flow originating inside the pruned part $S_j$, then the number of edges in the cut stays the same because it was placed there due to an edge between $S_j$ and $V \setminus A$. The total number of other crossing edges is at most twice the number of edges that cross in the residual graph because no more edges can carry flow out that doesn't originate inside. But these edges are therefore at most $\phi \vol_G(S_j)/(64 \log_2^2 n)$ in total. On the other hand, at least $\vol_G/(8\log_2^2 n)$ flow was absorbed. Since the initial flow is $\frac{2}{\phi} \cdot |E(A, V \setminus A)|$ we can charge a decline in flow of $2/\phi$ for each such edge and obtain our claimed bound. 
\end{proof}

We conclude the section with a proof of the main lemma. 

\begin{proof}[Proof of \Cref{lma:trimming}]
    The first item directly follows from \Cref{clm:expansion}, the second item directly follows from \Cref{clm:size_loss} and the third item follows from \Cref{clm:cut_sparse}. Therefore, it remains to argue that the work and span bounds are correct. By \Cref{clm:work_span_oracle}, the work and span bound is achieved for the individual oracle calls to $\textsc{ParallelUnitFlow}$. By \Cref{clm:outer_while_term}, there are at most $\log_2$ such calls, which adds a $\log_2 n$ factor to both work and depth. Since an individual step of ball growing in \Cref{lne:batch_pruning:while} can be implemented in parallel, determining the final set $S_j$ at iteration $i$ can be implemented in depth $\tilde{O}(1/\phi)$ and work $\tilde{O}(m)$. We recall that the main loop runs for at most $\log_2 n$ iterations by \Cref{clm:outer_while_term} and observe that all other operations can be implemented in parallel directly. This yields the desired bounds for work and span.   
\end{proof}

\section{Parallel Unit Flow}\label{sec:parallelimplementation}

In this section, we show \Cref{lma:parallelPushRelabel}, whose proof was deferred in \Cref{sec:trimming}. We restate the lemma for the readers convenience.

\parallelPushRelabel*

\paragraph{Algorithm Description. } The parallel unit-flow algorithm roughly follows the template already presented in \cite{GT88PushRelabel}, and we use the structure of our flow instance to show our work and span bounds. In particular, we exploit that every vertex is a sink proportional to its degree to relate the work to the amount of excess. Furthermore, we note that while our trimming algorithm operates on an undirected graph, unit-flow algorithm operates on a directed residual graph. See \Cref{alg:parallel-unit-flow} and \Cref{alg:push_then_relabel} for pseudo-code. 

We are given a directed graph $G$ with capacities $\cc$, a source vector $\DDelta$ and a sink vector $\nnabla\geq \gamma \cdot \deg_G$ where $\deg_G$ is an upper bound on both the in- and out-degree of $G$. At the start a zero flow $\ff_0 \gets \veczero$ and a level function $l: V \mapsto \{0, \ldots, h\}$ to $l(v) \gets 0$ for every vertex $v \in V$ are initialized. 

Our algorithm goes through stages $i = 1, \ldots, 8 \cdot \log_2 n$. At the $i$-th stage, we let $\nnabla_i \defeq \frac{i}{8 \log_2 n} \cdot \nnabla$, and let $x_i = \norm{\ex^G_{\ff_{i - 1}, \DDelta, \nnabla_{i - 1}}}_1$ denote the amount of excess flow that has not yet been routed. Then, the stage will attempt to compute a good unit-flow $\ff_i'$ until it either succeeds and the algorithm terminates, or the remaining excess not at level $h + 1$ is less than $\xx_i/2$ and the algorithm sets $\ff_i \gets \ff_{i - 1} + \ff_i'$ and proceeds with the next stage.  

To construct the flow $\ff'_i$, the algorithm runs the following algorithm on the residual graph $G$ with residual capacities $\cc_{\ff_{i - 1}}$. It initializes the source to $\ex^G_{\ff_{i - 1}, \DDelta, \nnabla_{i - 1}}$, sinks to $\nnabla/8 \log_2 n$. It then goes over the levels $j = h, \ldots, 1$ in a top to bottom order, and pushes all the flow from vertices at level $j$ to vertices at level $j-1$ until for every vertex $v$ such that $l(v) = j$ either there is no excess left, or all edges from $v$ to vertices at level $j - 1$ are saturated. Finally, after having processed each level, all vertices $v \in V$ for which all edges to vertices at level $l(v) - 1$ are saturated increase there level by one if they are not yet at level $h + 1$. Once the amount of excess not at level $h + 1$ dropped by a factor $2$ and we continue with the next stage as described above. 

Finally it returns the accumulated flow $\ff_{8 \log_2 n}$ and the level function $l$. 

\begin{algorithm}[H]
\caption{\textsc{ParallelUnitFlow}$(G, \cc, \DDelta, \nnabla, h)$}
\label{alg:parallel-unit-flow}
\DontPrintSemicolon
$\ff_0 \gets \veczero$; $\nnabla_0 \gets \veczero$: $\forall v \in V: l(v) = 0$\\
\For{$i = 1, \ldots, 8 \cdot \log_2 n$}{\label{lna:main_for} 
$x_i \gets \sum_{v \in V: l(v) \neq h + 1} \ex^G_{\ff_{i - 1},\DDelta,\nnabla_{i -1}}(v)$ \tcc*{Non settled excess}  $\ff_i' \gets \veczero$;
$\nnabla_i \gets \frac{1}{8 \log_2 n} \nnabla$ \\
\While(\label{lna:at-least-half}){$\sum_{v \in V: l(v) \neq h + 1}  \ex^G_{\ff_{i - 1}+\ff_i',\DDelta,\nnabla_i}(v) \geq x_i/2$}{
    $(\ff_i', l) \gets \textsc{PushThenRelabel}(G, \cc_{\ff_{i - 1}}, \ff_{i}', \ex^G_{\ff_{i - 1},\DDelta,\nnabla_{i -    1}}, \nnabla_i, h, l)$ \\
}
$\ff_i \gets \ff_{i-1}+\ff_i'$\\
}
$\forall v \in V$ s.t. $l(v) = h + 1$: $l(v) \gets h$ \\
\Return $(\ff_{8 \cdot \log_2 n}, l)$ \label{lne:final_return}
\end{algorithm}

\begin{algorithm}[H]
\caption{$\textsc{PushThenRelabel}(G, \cc, \ff, \DDelta, \nnabla, h, l)$}
\label{alg:push_then_relabel}
\DontPrintSemicolon
\For{j = h \ldots, 1}{
    In parallel, push all flow from all vertices $v$ with $l(v) = j$ that have excess flow to vertices $u$ with level $l(u) = j - 1$ until there either is no flow left or all the edges to such vertices are saturated. Update $\ff$ accordingly.  
}
For all vertices $v$ that only have saturated edges going to level $l(v) - 1$ and have no remaining sink capacity, increase their level $l(v) \gets \min(l(v) + 1, h + 1)$. \\
\Return $(\ff, l)$
\end{algorithm}

\paragraph{Proof of \Cref{lma:parallelPushRelabel}.} 

We first show that the algorithm runs in the desired work and span. 

\begin{claim}[Runtime] \label{clm:parallel_unit_flow_rt}
    Given $\nnabla \geq \gamma\cdot\deg_G$, $\DDelta(v) \leq \eta \cdot \deg_G(v)$ for all $v \in V$, and $\norm{\DDelta}_1 \leq 2m$ the algorithm $\textsc{ParallelUnitFlow}(G, \cc, \DDelta, \nnabla, h)$ described in \Cref{alg:parallel-unit-flow} can be implemented with total work $\tilde{O}(mh\eta/\gamma)$ and span $\tilde{O}(h^2\eta/\gamma)$ where $\eta \geq \norm{\cc}_{\infty}$. 
\end{claim}

\begin{proof}
    The outer loop of \Cref{alg:parallel-unit-flow} has $\tilde{O}(1)$ iterations by the definition of the algorithm. To show that the number of steps of the while loop at \Cref{lna:at-least-half} is also suitably bounded, we will refer to flow at vertices of level $h + 1$ as settled, and we observe that such flow never leaves the vertex again by the description of our algorithm. In every iteration of the loop at \Cref{lna:at-least-half},  every unit of unsettled excess gets both moved and raised by a  level.  Every individual vertex $v$ can raise at most $\eta\cdot \deg_G(v) (h + 1)$ units of excess flow throughout. This is because it can contain at most $\eta\cdot\deg_G(v)$ excess at any point since the capacities are bounded by $\eta$, and it is raised at most $h + 1$ times. But every vertex also absorbs $\gamma \deg_G(v)/(8 \log_2 n)$ flow before it is every raised. Therefore the total units of flow that raise in level can be at most $O(x_i \cdot \eta \cdot h \cdot \frac{\log_2 n}{\gamma})$ before all the flow is settled in iteration $i$ of the main loop at \Cref{lna:main_for}. But since each iteration of the inner while loop at \Cref{lna:at-least-half} raises at least $x_i/2$ units of flow until half the flow is settled, the total number of iterations until half the flow is settled are bounded by $O( \eta \cdot h \cdot \frac{\log_2 n}{\gamma})$. Finally each call to \Cref{alg:push_then_relabel} causes another multiplicative factor of $h$ in the span by the description of the algorithm. Therefore, the total span of the algorithm is $\tilde{O}(h^2 \eta/\gamma)$.
    
    Since all work can be charged to flow going down one level the total work of the algorithm is $\tilde{O}(m  h  \eta/\gamma)$ as claimed, again by the fact that the total increase in level is bounded by $O(x_i \cdot \eta \cdot h \cdot \frac{\log_2 n}{\gamma})$ and $x_i \leq \norm{\DDelta}_1 \leq 2m$. 
    
    This concludes the runtime analysis since all other parts of the algorithm can directly be implemented in work $\tilde{O}(m)$.
\end{proof}

Then, we show correctness. 

\begin{claim}[Correctness] \label{clm:parallel_unit_flow_cor}
    Given $\norm{\DDelta}_1 \leq 2m$, the algorithm \Cref{alg:parallel-unit-flow} computes a flow $\ff$ and labeling $l$ such that
    \begin{enumerate}[label=(\roman*)]
    \item If $l(u)>l(v)+1$ where $\{u,v\}$ is an edge, then $\{u,v\}$ is saturated in the direction from $u$ to $v$, i.e. $\ff(u,v)=\cc(u,v)$.
    \item If $l(u)\geq 1$, then $u$'s sink is nearly saturated, i.e. $\ff(u)\geq \nnabla(u)/(8 \cdot \log_2 n)$.
    \item[(iii)] If $l(u)<h$, then there is no excess mass at $u$, i.e. $\ex_{\Delta,\nabla,f}(u) = 0$.
\end{enumerate}
\end{claim}
\begin{proof}
     We show the three items separately. 
    \begin{enumerate}[label=(\roman*)]
    \item The first item is maintained as an invariant throughout the execution of the algorithm, which follows from the description of the algorithm in \Cref{alg:push_then_relabel} since we only ever increase the levels of vertices that only have saturated edges to vertices of lower level. 
    \item The second item also directly follows from the description of the algorithm since for a vertex to be asssigned a non-zero level its, sink has to be saturated for at least one iterate $i$.
    \item[(iii)] Since the amount of flow that is not yet settled reduces by a factor of two in each round, there is no non-settled flow left over after $8\log_2 n$ rounds because $2m/2^{8\log_2 n} < 1$ which directly implies the third item since all settled flow is either routed or at level $h$ after the algorithm terminates.
\end{enumerate}
This concludes our proof. 
\end{proof}

We conclude with the proof of the main lemma. 

\begin{proof}[Proof of \Cref{lma:parallelPushRelabel}]
    Follow directly from \Cref{clm:parallel_unit_flow_rt} and \Cref{clm:parallel_unit_flow_cor}.
\end{proof}

\newpage
\bibliographystyle{alpha}
\bibliography{refs}

\appendix
\section{Proof of \Cref{thm:mainResult}}
\label{sec:proofOfMain}

\subsection{Expander Decompositions by Recursing on Balanced Cuts}

In this section, we show how to derive \Cref{thm:mainResult} by parallelizing the framework given by \cite{saranurak2019expander} that was also previously used in the parallel algorithm in \cite{chang2020deterministic}.

The missing key component for this framework is a parallel implementation of the cut-matching game algorithm from \cite{khandekar2009graph}. In the next section, we give such an implementation by speeding up an adaption by \cite{chang2020deterministic} with our techniques from \Cref{sec:parallelimplementation}.

\begin{theorem}\label{thm:paralllelCutMatching}
Let $G = (V, E)$ be an $m$-edge graph, and let $0 < \phi < 1$ be any parameter. There is a randomized algorithm with work $\tilde{O}(m/\phi^2)$ and span $\tilde{O}(1/\phi^{4})$ that finds a cut $C \subseteq V$ satisfying $0 \leq \vol_G(C) \leq m$ and $|E_G(C, V \setminus C)| < \frac{\phi \cdot m}{64 \log^2 n}$ and if $\vol_G(C) \leq m/100$ then we also have that $G[V \setminus C]$ is a $\tilde{\Omega}(\phi)$-nearly expander. The algorithm succeeds with high probability.
\end{theorem}

Given this algorithm and our trimming procedure from \Cref{lma:trimming}, we can now straightforwardly compute an expander decomposition via the algorithm \textsc{ComputeExpDecomp}$(G, \phi)$ presented in \Cref{alg:expanderDecomp}. The algorithm either finds a balanced sparse cut and then recurses on both sides of the cuts separately, or it identifies a large expander and only needs to recurse on the (smallish) remainder of the graph.

\begin{algorithm}
$C \gets \textsc{Cut-Matching}(G,\phi)$.\\
\If{$\vol_G(C) > m/100$}{
    \Return $\textsc{ComputeExpDecomp}(G[C], \phi) \cup \textsc{ComputeExpDecomp}(G[V \setminus C], \phi)$.
} \Else{
    $A' \gets \textsc{Trimming}(G, V \setminus C, \phi)$.\\
    \Return $\{A'\}\cup \textsc{ComputeExpDecomp}(G[V \setminus A'], \phi)$.\label{lne:placeAprime}
}
\caption{\textsc{ComputeExpDecomp}$(G, \phi)$}
\label{alg:expanderDecomp}
\end{algorithm}

The proof of \Cref{thm:mainResult} is now rather straightforward. 

We first observe that the recursion depth of the algorithm is $O(\log n)$ since if the algorithm enters the if-statement the subgraph $G[C], G[V \setminus C]$ both have at most $(1-1/100)m$ edges by \Cref{thm:paralllelCutMatching}. If the algorithm enters the else-statement, it only recurses on $G[V \setminus A']$ which consists of at most $\frac{51}{100} \cdot m$ edges as can be derived from combining the guarantees of \Cref{thm:paralllelCutMatching} and \Cref{lma:trimming}. This yields immediately that the runtime is $\tilde{O}(m/\phi^2)$ and span is $\tilde{O}(1/\phi^4)$. 

For correctness, it suffices to see that each set $A'$ that is finally put into the expander hierarchy is placed there in \Cref{lne:placeAprime}. It can be seen from \Cref{thm:paralllelCutMatching} and \Cref{lma:trimming} that each such set has the property that $G[A']$ is an $\tilde{\Omega}(\phi)$-expander. To bound the total error, observe that the error is upper bounded by the number of edges leaving the clusters on which we recurse in each step. But in each invocation of $\textsc{ComputeExpDecomp}(G, \phi)$ on an $m$-edge graph there are at most $O(\phi m)$ edges crossing this cut. Since the number of edges inputted to all such invocations is $O(m \log n)$ by the bound on the recursion depth and the fact that we recurse on disjoint subgraphs, we obtain an error bound of $\tilde{O}(\phi m)$, as desired.

Finally, we observe that since \Cref{thm:paralllelCutMatching} succeeds with high probability, and is only called at most $n$ times, we have that the algorithm succeeds with high probability overall.

\subsection{Implementing the Cut-Matching Algorithm}

In this section, we prove \Cref{thm:paralllelCutMatching}. We note that this theorem is obtained by parallelizing the cut-matching algorithm from \cite{khandekar2009graph}. This was previously done also in the work of  \cite{chang2020deterministic} (see Lemma 3.1), however, their proof loses heavily in terms of work, span, and error. 

However, following the proof of Lemma 3.1 in \cite{chang2020deterministic}, one can verify that the only bottleneck is the computations of the flows in the cut-matching problem. In each of the $O(\log^2 n)$ iterations of the cut-matching algorithm that is used internally, one has to solve a flow problem that puts at most one unit of source or one unit of sink at each vertex and were edge capacities are again of order $\tilde{O}(1/\phi)$. 

But we observe that our algorithm from \Cref{sec:parallelimplementation} can also be used to solve this flow problem efficiently. We can prove the following theorem which tightens the key flow result from \cite{chang2020deterministic} to achieve much better work, span and error. We defer the proof to 
\Cref{subsec:keyFlowLemma}.

\begin{restatable}[Alternative cut-or-match, see Lemma D.7 of \cite{chang2020deterministic}]{theorem}{keyFlowLemma} \label{thm:d7}
Consider a graph $G = (V, E)$ with degree bounded by $16$ and a parameter $0 < \phi < 1$. Given a set of source vertices $S$ and a set of sink vertices $T$ with $|S| \leq  |T|$, there is
an algorithm that finds a cut $C$ and a set of $S$-$T$ paths $\mathcal{P}$ embedding a (possibly non-perfect) matching $M$ between $S$ and $T$ such that:
\begin{itemize}
    \item every edge in $G$ appears on at most $\tilde{O}(1/\phi)$ paths in $\mathcal{P}$, and 
    \item $C$ contains all vertices in $S$ that are not the starting vertex of a path in $\mathcal{P}$, and
    \item $C$ does not contain any vertex in $T$ that is not the last vertex of a path in $\mathcal{P}$, and
    \item either $C = \emptyset$, or $|E_G(C, V \setminus C)| \leq \phi \cdot \vol_G(C) + 2 \cdot \phi \cdot m$.
\end{itemize}
The algorithm requires $\tilde{O}(m/\phi^2)$ work and span $\tilde{O}(1/\phi^4)$.
\end{restatable}

Using the above flow lemma in-lieu of Lemma D.7 of \cite{chang2020deterministic} in the proof of Lemma 3.1 in \cite{chang2020deterministic}, we derive a parallel cut-matching algorithm that requires us to run the above algorithm $\tilde{O}(1)$ times and implements the rest of the framework in $\tilde{O}(m/\phi)$ work and $O(1/\phi)$ span. We point out that \Cref{thm:d7} has an additive term of $2\phi m$ on the right-hand side of the inequality in the last bullet point, that is not present in Lemma D.7. However, it can be verified in the proof of Lemma 3.1 in \cite{chang2020deterministic} that this additive term does not change any of the claim guarantees asymptotically.

We note that technically speaking, the Lemma 3.1 in \cite{chang2020deterministic} is only proven for bounded-degree graphs, however, a transformation between general graphs and bounded-degree graphs that lifts this result (at the loss of constant factors) is folklore and can be implemented with $\tilde{O}(m)$ work and $\tilde{O}(1)$ span (see for example \cite{chen2023simple}).

\subsection{Implementing the Alternative Cut-Or-Match Lemma}
\label{subsec:keyFlowLemma}

Finally, we prove the following lemma using our techniques from \Cref{sec:parallelimplementation}.

\keyFlowLemma*

We prove \Cref{thm:d7} via a simple adaptation of the algorithm presented in \Cref{sec:parallelimplementation}. See \Cref{alg:parallel_matching} for pseudo-code. Our algorithm initializes a flow instance with source $\vecone_S$ and sink $\vecone_T$. The algorithm roughly follows the template of the trimming algorithm described in \Cref{sec:trimming} and \Cref{alg:parallel_matching}. We again introduce a level function for $h \defeq 100 \cdot \log_2 (m)/\phi$ pointing to levels $0, \ldots, h + 1$, and we say flow is settled when it is at level $h + 1$ since we never move it from there. Then, we run a parallel unit flow algorithm as long as $\phi \cdot m$ flow remains not settled, and we show that this implies that every round makes a lot of progress. Finally, once most of the flow is settled, we output the set $C$ as a combination of the origin of all unsettled flow, and a sparse cut in the level graph which we show to be a sparse cut in the original graph as well. 

\begin{algorithm}[H]
\caption{\textsc{ParallelMatching}$(G, S, T)$}
\label{alg:parallel_matching}
\DontPrintSemicolon
$\ff \gets \veczero$; $\forall v \in V: l(v) \gets 0$; $\cc \gets   \frac{2}{\phi} \cdot \vecone$; $h \gets 100 \cdot \log_2 (m)/\phi$  \\
\While{$\sum_{v \in V: l(v) \neq h + 1}  \ex^G_{\ff, \vecone_S, \vecone_T}(v) \geq \phi \cdot m / 16$}{\label{alg:parallel_matching:while1}
    $\tt \gets \max(-\BB^T \ff - \vecone_S + \vecone_T ,\veczero)$
    $(\ff, l) \gets \textsc{PushThenRelabel}(G, \cc_{\ff_{i - 1}}, \ff, \ex^G_{\ff,\vecone_S,\vecone_T}, \tt , h, l)$ \\
}
$\forall v \in V$ such that $l(v) = h + 1$: $l(v) \gets h$. \\
$j \gets 0$ \\
\If{$S_0 \neq \emptyset$}{
$S_0 \gets \{v \in V: l(v) = h\}$; \\
\While{$|E_{\ff}(S_j, V \setminus S_j)| \geq \frac{\phi}{4} \vol_G(S_j)$}{ \label{alg:parallel_matching:while2}
    $j \gets j + 1$; $S_j \gets \{v \in V: l(v) \geq h + 1 - j\}$
}
}\Else{
    $S_j \gets \emptyset$
}
Let $\mathcal{P}$ be the path decomposition of all but a $\phi$ fraction of $\ff$. \\
Let $U \subseteq S$ be the set of all the vertices that don't have a path leaving in $\mathcal{P}$. \\
$C \gets S_j \cup U$\\
\Return $(\mathcal{P}, C)$ 
\end{algorithm}

\begin{proof}[Proof of \Cref{thm:d7}]
    We first show that the while loop at \Cref{alg:parallel_matching:while1} terminates in $\tilde{O}(1/\phi^3)$ iterations. Initially there are at most $m$ units of unsettled excess in the graph, where we will refer to excess at level $h + 1$ as settled since it never gets routed away again. We notice that each round of $\textsc{PushThenRelabel}()$ increases the level of every unit of unsettled flow by one. But, every vertex can increase at most $\tilde{O}(1/\phi)$ units of flow whenever it increases its level, and therefore the total amount of flow increase is at most $O(m/\phi^2)$. Since every iteration increases at least $\phi m$ units of flow, this process converges after $\tilde{O}(1/\phi^3)$ iterations. Therefore, the depth of the while loop at \Cref{alg:parallel_matching:while1} is at most $\tilde{O}(\phi^{-4})$ and the total work is $\tilde{O}(m/\phi^2)$ since each iteration of $\Cref{alg:push_then_relabel}$ can be implemented in $\tilde{O}(\phi^{-1})$ depth and we can pay for work with a level decrease, and there are at most $h$. 

    We then consider the second while loop. We show that this loop has to converge in less than $h$ iterations. Since the set $S_0$ is not empty and the graph is connected, we have that $\vol_G(S_0) \geq 1$, and because the residual graph contains all edges between $S_j$ and $S_{j + 1}$ which follows directly from the description of the $\textsc{PushThenRelabel}()$ routine described in \Cref{alg:push_then_relabel}, we have that $\vol_G(S_{j + 1}) \geq (1 + \phi/4) \vol_G(S_{j})$ and therefore $\vol_G(S_{j + 1}) > m$ which is a contradiction. Therefore, this part can be implemented with total work $\tilde{O}(m)$ and depth $\tilde{O}(\phi^{-1})$ as well. 
    
    If we were to compute the flow decomposition naively, we would obtain a congestion bound of $\tilde{O}(\phi^{-2})$ since an edge can be routed through in opposing directions up to $h$ times. But we can remedy this issue by arbitrarily pairing up paths whenever they cross in both directions over an edge and swapping their tails. This can cause some paths to be very long, but at most a $\phi$ fraction of the $m$ paths can be longer than $\tilde{O}(\phi^{-2})$ by Markov's inequality. Therefore, we can just explore all paths in parallel and drop the paths that are too short. 
    
    Finally, we show a bound on the cut $E(C, V\setminus C)$. The $U$ component of $C$ has volume at most $2\phi m$ since there is at most $\phi m$ unsettled flow and at most $\phi m$ paths that are dropped, and can therefore never pose an issue. Therefore, we just bound the size of the cut $E(S_j, V \setminus S_j)$. We notice that the size of the cut in the residual graph is at most $E_{\ff}(S_j, V \setminus S_j) \leq \frac{\phi}{4} \vol_G(S_j)$ when the algorithm terminates. All the other edges are saturated, but the total flow inside $S_j$ is at most  $\frac{1}{2} \vol_G(S_j)$ that could get routed in and $\vol_G(S_j)$ flow that originated inside. But this amount of flow can saturate at most $\frac{3}{2} \cdot \frac{\phi}{2} \cdot \vol_G(S_j) = \frac{3\phi}{4} \vol_G(S_j)$. Therefore, $|E(S_j, V \setminus S_j) \leq \phi \cdot \vol_G(S_j)$ as desired. 
    This concludes the proof 
\end{proof}

\section{Review of Algorithms for Expander Decompositions} \label{sec:review_exp_decomp}

\paragraph{Sequential Algorithms for Expander Decompositions.} Expander decompositions were first introduced as a graph clustering framework in \cite{VempalaV00clustering}. The first nearly-linear time sequential algorithm for computing nearly-expander decompositions was presented in \cite{st:exp_dec} as part of their seminal result on solving Laplacian linear equations \cite{spielman2004nearly}. While sufficient for their purposes, the algorithm of \cite{st:exp_dec} had the defect that clusters were only guaranteed to be contained in a possibly larger unknown expander graph, i.e. nearly-expanders. This defect was remedied by the works of \cite{wulff2017fully, nanongkai2017dynamic, nanongkai2017dynamicMinimum}, but they instead suffered an almost-linear runtime meaning that the sub-polynomial overhead of their algorithm is larger than poly-logarithmic and a subpolynomial loss in the error of the expander decomposition. Finally, \cite{saranurak2019expander} presented a sequential algorithm that runs in time $\tilde{O}(m/\phi)$, achieves quality $\tilde{O}(\phi m)$ and ensures that each component is a $\phi$-expander. All of the above results however were randomized due to the internal use of the cut-matching algorithm \cite{khandekar2009graph}. In \cite{chuzhoy2020deterministic}, a deterministic cut-matching algorithm was given, effectively derandomizing all previously mentioned algorithms, however, again at the expense of subpolynomial losses in overhead and quality.

\paragraph{Parallel and Distributed Algorithms for Expander Decompositions.}
In the parallel and distributed setting,  \cite{chang17distexp} showed that a variant of expander decompositions can be computed quickly, which allowed them to get improved triangle detection algorithms.  
This culminated in the state of the art parallel algorithms in \cite{chang2019improved, chang2020deterministic}, which nearly achieve the bound of \cite{saranurak2019expander} but lose a third root in the quality.

\end{document}